\documentclass[pra,twocolumn, preprintnumbers,amsmath,amssymb,superscriptaddress]{revtex4}

\usepackage{graphicx, color, graphpap}% Include figure files
\usepackage{amsmath}
\usepackage{enumitem}
\usepackage{amssymb}
\usepackage{amsthm}
\usepackage{pstricks}
\usepackage{float}
\usepackage{multirow}
\usepackage{hyperref}

\long\def\ca#1\cb{} %Use for commenting out: \ca...\cb

\newcommand{\ket}[1]{|#1\rangle}               %ket
\newcommand{\colo}{\,\hbox{:}\,}              %colon in math with less space
\newcommand{\bra}[1]{\langle #1|}              %bra
\newcommand{\dya}[1]{\ket{#1}\!\bra{#1}}
\newcommand{\dyad}[2]{\ket{#1}\!\bra{#2}}        %dyad
\newcommand{\ip}[2]{\langle #1|#2\rangle}      %the inner product
\newcommand{\fid}{\text{fid}}

\newcommand{\Sbb}{\mathbb{S}}

\newcommand{\EC}{\mathcal{E}}

\newcommand{\HC}{\mathcal{H}}
\newcommand{\IC}{\mathcal{I}}

\newcommand{\RC}{\mathcal{R}}

\newcommand{\Tr}{{\rm Tr}}

\renewcommand{\geq}{\geqslant}
\renewcommand{\leq}{\leqslant}
\newcommand{\mte}[2]{\langle#1|#2|#1\rangle }

\newcommand{\ot}{\otimes}
\newcommand{\ad}{^\dagger}
\newcommand*{\id}{\openone}

\newcommand{\al}{\alpha }
\newcommand{\bt}{\beta }

\newcommand{\lm}{\lambda }
\newcommand{\Lm}{\Lambda }

\newcommand{\sg}{\sigma }

\newcommand{\om}{\omega }

\newcommand{\ddd}{d }

\newtheoremstyle{example}{\topsep}{\topsep}%
{}%         Body font
{}%         Indent amount (empty = no indent, \parindent = para indent)
{\bfseries}% Thm head font
{.}%        Punctuation after thm head
{   }%     Space after thm head (\newline = linebreak)
{\thmname{#1}\thmnumber{ #2}}%\thmnote{ #3}}%         Thm head spec
\theoremstyle{example}
\newtheorem{theorem}{Theorem}
\newtheorem{lemma}[theorem]{Lemma}
\newtheorem{corollary}[theorem]{Corollary}

\newtheorem{proposition}[theorem]{Proposition}

\theoremstyle{definition}

\begin{document}

\title{Complementary sequential measurements generate entanglement}

\author{Patrick J. Coles}
\affiliation{Centre for Quantum Technologies, National University of Singapore, 2 Science Drive 3, 117543 Singapore}
%\date{\today}

\author{Marco Piani}
\affiliation{Institute for Quantum Computing and Department of Physics and Astronomy, University of Waterloo, N2L3G1 Waterloo, Ontario, Canada}

\begin{abstract}
We present a new paradigm for capturing the complementarity of two observables. It is based on the entanglement created by the interaction between the system observed and the two measurement devices used to measure the observables sequentially. {Our main result is a lower bound on this entanglement and resembles well-known entropic uncertainty relations.} Besides its fundamental interest, this result directly bounds the effectiveness of sequential bipartite operations---corresponding to the measurement interactions---for entanglement generation. We further discuss the intimate connection of our result with two primitives of information processing, namely, decoupling and coherent teleportation.
\end{abstract}


\maketitle

Heisenberg's uncertainty principle~\cite{price1977uncertainty} tries to capture one of the fundamental traits of quantum mechanics: the complementarity of observables like position and momentum. There are several variants of the principle which may be considered conceptually very different~\cite{ozawa_pra_2003}. For example, one can consider the uncertainty related to the independent measurement of two observables, with the measurements performed on two \emph{independent but identically prepared} quantum systems.  In this scenario, the uncertainty principle for complementary observables can be understood as stating that there is an \emph{unavoidable uncertainty} about the outcomes of the associated measurements. Alternatively, one can consider the \emph{sequential} measurement of such two observables, performed on the \emph{same} physical system. In this case, the uncertainty principle is  understood as the \emph{unavoidable disturbance} on the second observable due to the measurement of the first. Although this latter disturbance-based interpretation of the principle is the one originally considered by Heisenberg in his famous $\gamma$-ray thought experiment~\cite{Heisenberg}, researchers have more often focussed on the first scenario.

Unavoidable uncertainty was stated quantitatively by Kennard~\cite{kennard1927quantum} and Robertson~\cite{Robertson} in the famous uncertainty relation involving standard deviations. Since then, uncertainty relations have been cast in information-theoretic terms \cite{EURreview1}. For example, a well-known \emph{entropic uncertainty relation} is that of Maassen and Uffink \cite{MaassenUffink}. Working in finite dimensions, they consider two orthonormal bases $\{ \ket{X_j} \}$ and $\{ \ket{Z_k} \}$ for the Hilbert space $\HC_S$ of a quantum system $S$, to which one can associate observables $X$ and $Z$, respectively. For any state $\rho_S$, they find
\begin{equation}
H(X)+H(Z)\geq \log (1/c),
\label{eqn1}
\end{equation}
where $H(X):=-\sum_j p(X_j) \log p(X_j)$ is the Shannon entropy associated with the probability distribution $p(X_j):=\bra{X_j} \rho_S \ket{X_j}$ (similarly for $H(Z)$), logarithms are taken in base $2$, and $c := \max_{j,k} |\ip{X_j}{Z_k}|^2 $ quantifies the complementarity between the $X$ and $Z$ observables. The r.h.s.\ of \eqref{eqn1} vanishes when $X$ and $Z$ share an eigenstate. At the other extreme, when $X$ and $Z$ are complementary---so-called mutually unbiased bases (MUBs) with $|\ip{X_j}{Z_k}|^2 = 1/\ddd $, $\forall j,k$, and $\ddd  = \dim (\HC_S)$---the r.h.s.\ becomes $\log \ddd $. In the latter case, Eq.~(\ref{eqn1}) implies that when our uncertainty about $X$ approaches zero, our uncertainty about $Z$ must approach its maximum value $\log \ddd $.

In this Letter, we offer a novel view on what complementarity entails by relating it to another fundamental trait of quantum mechanics: \emph{entanglement}~\cite{HHHH09}. In~\cite{ColesCollapsePRA2012} it was already proved that  an entropic uncertainty relation like (\ref{eqn1}) has a correspondent \emph{entanglement certainty relation}. In more detail, Ref.~\cite{ColesCollapsePRA2012} considers the generation of entanglement between measurement devices and independent, although identically prepared, copies of some physical system, and proves that, when dealing with complementary observables, there is \emph{unavoidable creation of entanglement} between at least one copy of the system and one measuring device. Here, as Heisenberg did originally, we instead consider sequential measurements performed on the same physical system, rather than  independent copies of the system; on the other hand, following~\cite{StrKamBru11,PianiEtAl11,GharEtAl2011,PianiAdessoPRA.85.040301,ColesCollapsePRA2012}, we still focus on the entanglement generated between the system and the measurement devices. In general, for any $X$ and $Z$, we can lower-bound the entanglement $E(X,Z)$ between the system and the measurement devices created from sequentially measuring $X$ and $Z$ with
 \begin{equation}
\label{eqn2}
E(X,Z) \geq \log (1/c),
\end{equation}
where the $c$ factor appearing here is \textit{precisely the same $c$ appearing in Eq.~\eqref{eqn1}}, and we provide more details on how we quantify entanglement in the following.

Besides the fact that our approach connects in a fundamental way two basic properties of quantum mechanics, complementarity---in the sequential-measurement scenario---and entanglement, our results have also direct operational interpretations. On one hand, they provide bounds on the usefulness of sequential bipartite operations---corresponding to the measurement interactions---for entanglement generation. On the other hand, we argue below that our analysis is directly linked to the quantum information processing primitives of decoupling~\cite{SchuWestErrCorr, DupEtAl2010, DupuisThesis, GroPopWinPRA72.032317, BuscemiNJP2009} and coherent teleportation~\cite{Brassard1998,HarrowPRL2004}. 

\bigskip

\emph{Setup.---}The basic setup corresponding to our main result is given in Fig.~\ref{fgr1}. The system is initially described by some arbitrary density operator $\rho^{(0)}_S$. It first interacts with a device $M_1$ meant to measure the observable $X$. We depict this interaction with the controlled-NOT (CNOT) symbol, although more generally it represents a controlled-shift unitary, $U_X = \sum_j [X_j] \ot \Sbb^j \ot \id_{M_2}$, acting on the tripartite Hilbert space $\HC_{SM_1M_2}$, where $\Sbb = \sum_k \dyad{k+1}{k}$ is the shift operator and $[X_j]$ is a shorthand notation for the dyad $\dya{X_j}$. This is a unitary model for the measurement process \cite{ZurekReview}. After this, the system interacts with a second device $M_2$, which measures the $Z$ observable; the unitary is given by $U_Z = \sum_j [Z_j] \ot \id_{M_1} \ot \Sbb^j$. We suppose that both $M_1$ and $M_2$ are initially in the $\ket{0}$ state, although later in the article we consider the effect of relaxing this assumption. We denote the states at times $t_0$, $t_1$, and $t_2$ in Fig.~\ref{fgr1} as $\rho^{(0)}_{SM_1M_2}$, $\rho^{(1)}_{SM_1M_2}$, and $\rho^{(2)}_{SM_1M_2}$, respectively.

\bigskip

\emph{Entanglement generation.---}We focus on the bipartite entanglement $E(X,Z)$ between $S$ and the joint system $M_1M_2$ present in the final state
\[
\rho^{(2)}_{SM_1M_2} = \sum_{j,k,l,m} [Z_l][X_j]\rho^{(0)}_S[X_k][Z_m] \ot \dyad{j}{k}\ot \dyad{l}{m}.
\]
For concreteness we consider $E$ to be the distillable entanglement~\cite{HHHH09}, i.e., the optimal rate for distilling Einstein-Podolsky-Rosen (EPR) pairs $(\ket{0}\ket{0}+\ket{1}\ket{1})/\sqrt{2}$ using local operations and classical communication (LOCC) in the asymptotic limit of infinitely many copies of the state. However, our result holds for several other entanglement measures, because distillable entanglement is itself a lower bound for such measures~\cite{HHHH09}.

Consider first the case where $X$ and $Z$ are MUBs. In this case, $\rho^{(2)}_{SM_1M_2}$ is maximally entangled across the $S$:$M_1M_2$ cut, \textit{regardless of the system's initial state $\rho^{(0)}_S$}. One can see this by noting that, if we choose the LOCC operation that measures  $M_1$ in the standard basis and communicates the result to the party holding $S$, the resulting conditional pure state on $SM_2$ is, up to an irrelevant local change of basis, a maximally entangled e-dit of the form $\sum_{i=0}^{\ddd -1}\ket{i}\ket{i}/\sqrt{\ddd }$. Alternatively, and more elegantly, we can factor out a maximally entangled state simply by performing a local unitary on $M_1M_2$; more precisely, the following holds.
\begin{proposition}
\label{prp1}
Let $X$ and $Z$ be MUBs. Define $H_{M_1} = \sum_j \dyad{X_j}{j}$ and the controlled unitary $U_{M_1M_2} = \sum_j \sg_X^j \ot [j]$, where $\sg_X^j := \sqrt{\ddd } \sum_k \ip{X_k}{Z_j} [X_k]$. Then
\begin{equation}
U_{M_1M_2} H_{M_1} \rho^{(2)}_{SM_1M_2} H_{M_1}\ad U_{M_1M_2}\ad = [\Phi]_{SM_2} \ot (\rho^{(0)}_S)_{M_1},
\end{equation} 
with $\ket{\Phi}=(\sum_j \ket{Z_j}\ket{j})\sqrt{\ddd }$: the local unitary $U_{M_1M_2} H_{M_1}$ applied to $\rho^{(2)}_{SM_1M_2}$ leaves $M_1$ in the system's initial state $\rho^{(0)}_S$, and $SM_2$ maximally entangled.
\end{proposition}

\begin{figure}[t]
\begin{center}
\includegraphics[width=2in]{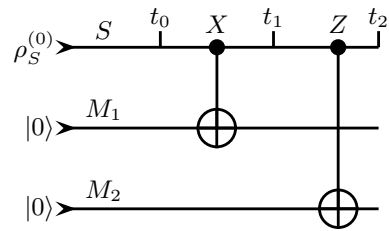}
\caption{Circuit diagram for the sequential measurement of the $X$ and $Z$ observables on system $S$. 
\label{fgr1}}
\end{center}
\end{figure}

Thus, in the case of MUBs, we can identify several tasks that are accomplished by sequentially measuring $X$ and $Z$ as in Fig.~\ref{fgr1}. Besides producing maximal entanglement, the state $\rho^{(0)}_S$ is ``teleported'' from the system to the measurement devices. Indeed, the protocol we have described above is commonly known as \emph{coherent teleportation}~\cite{Brassard1998,HarrowPRL2004}. Furthermore, since $S$ is maximally entangled to $M_1M_2$ at the end of the protocol, then, by the monogamy principle~\cite{seevinck2010}, $S$ must be completely uncorrelated with any other system $S'$. The procedure of performing an operation on $S$ to destroy its potential correlations with $S'$ is  known as \emph{decoupling}~\cite{SchuWestErrCorr, DupEtAl2010, DupuisThesis, GroPopWinPRA72.032317, BuscemiNJP2009}. Our main contribution is to extend the above discussion to the case where $X$ and $Z$ have partial complementarity ($c>1/\ddd $): \emph{Can we still create entanglement, coherently teleport, and decouple even if $X$ and $Z$ are not MUBs, and if so, to what degree?} 

Our main result \eqref{eqn2}, says that, as soon as  there is partial complementarity between $X$ and $Z$, some distillable entanglement is present in $\rho^{(2)}_{SM_1M_2}$.
\begin{theorem}
\label{thmMain230}
Let $E(X,Z)$ denote the distillable entanglement between $S$ and $M_1M_2$ at time $t_2$ in Fig.~\ref{fgr1}. Then \eqref{eqn2} holds.
\end{theorem}
\begin{proof}
We give two alternative proofs. The first is based on the uncertainty principle with quantum memory \cite{BertaEtAl} and the second is based on the monotonicity of entanglement under LOCC~\cite{HHHH09}. The second proof approach yields a slightly stronger version of \eqref{eqn2}.

In the first approach we we apply the uncertainty principle with quantum memory \cite{BertaEtAl} at time $t_1$ (just after the $X$ measurement) to get:
\begin{equation}
\label{eqn3}
H(X | M_1M_2)_{\rho^{(1)}} + H(Z | S')_{\rho^{(1)}} \geq \log(1/c)
\end{equation}
where we let $S'$ purify the initial state $\rho^{(0)}_S$, and where the first and second terms in \eqref{eqn3} are the conditional entropies of $\rho^{(1)}_{XM_1M_2} := \sum_j [X_j]\rho^{(1)}_{SM_1M_2}[X_j]$ and $\rho^{(1)}_{ZS'} := \sum_k [Z_k]\rho^{(1)}_{SS'}[Z_k]$ respectively. The von Neumann conditional entropy of $\sg$ is defined as $H(A|B)_{\sg} := H(\sg_{AB})-H(\sg_B)$, with $H(\sg) = - \Tr(\sg \log \sg)$ the von Neumann entropy. Because $X$ was already measured by $M_1$, we have $H(X | M_1M_2)_{\rho^{(1)}}=0$. Also, from a result in \cite{ColesDecDisc2012, ColesCollapsePRA2012}, we have $H(Z|S')_{\rho^{(1)}} = E(X,Z)$, completing the proof.

In the second approach, we note that the final entanglement is larger than the average entanglement obtained from measuring $M_1$ in the standard basis followed by communicating the result to the party holding system $S$. That is, $E(X,Z)\geq \sum_j p_j H(\rho^{(2)}_{S,j})$, where we used that the conditional states associated with different measurement outcomes are bipartite pure states, $p_j \rho^{(2)}_{SM_2,j} = \Tr_{M_1}[(\id \ot \dya{j}\ot \id)\rho^{(2)}_{SM_1M_2}]$, hence their entanglement is the entropy of the reduced state $\rho^{(2)}_{S,j} = \Tr_{M_2} (\rho^{(2)}_{SM_2,j})$. We obtain
\begin{equation}
\label{eqn4}
E(X,Z)\geq \sum_j p_j H( \{ |\ip{X_j}{Z_k}|^2 \}_k ),
\end{equation}
where the entropy on the r.h.s.\ is the classical entropy of the set of overlaps obtained from varying the index $k$. Equation~\eqref{eqn4} is slightly more complicated than \eqref{eqn2} because it depends on the initial state through the probabilities $p_j = \mte{X_j}{\rho^{(0)}_S}$. On the other hand, it is slightly stronger, implying \eqref{eqn2} by noting that Shannon entropy upper-bounds the min-entropy $H_{\min}(\{ q_k \})= - \log \max_k q_k$, and averaging over $j$ in \eqref{eqn4} yields a larger value than minimizing over $j$, completing the proof.
\end{proof}

So, even for limited complementarity, the circuit in Fig.~\ref{fgr1} still generates entanglement ``efficiently''. Using our main result, we also prove below that decoupling and coherent teleportation are approximately achieved in the case of approximate complementarity. We further consider two generalizations of our results: to the case of mixed measurement devices, and to the case of an arbitrary number of sequential measurements.

\bigskip

\emph{Decoupling.---}Decoupling \cite{SchuWestErrCorr, DupEtAl2010, DupuisThesis, GroPopWinPRA72.032317, BuscemiNJP2009} consists in transforming an arbitrary bipartite state $\rho_{SS'}$ into some tensor product $\sg_S \ot \sg_{S'}$, and it has specific applications in state merging~\cite{horodecki2005partial} and quantum cryptography~\cite{RevModPhysCrypto}. Decoupling strategies often involve a local operation performed on system $S$ only. Note that the effect on $S$ of the circuit of Fig.~\ref{fgr1} is equivalent to a random unitary channel $\rho^{(0)}_S \mapsto  (1/\ddd ^2) \sum_{k,l}(\sg_Z^k\sg_X^l) \rho^{(0)}_S (\sg_Z^k\sg_X^l)\ad$, consisting of $\ddd ^2$ unitaries each of which is a product of generalized Pauli operators, $\sg_X = \sum_j \om^j \dya{X_j}$ and $\sg_Z = \sum_j \om^j \dya{Z_j}$ with $\om = e^{2\pi i / \ddd }$. It is well-known that when $X$ and $Z$ are MUBs  this results in $\rho^{(0)}_{SS'}\mapsto \rho^{(2)}_{SS'}=\id/\ddd  \ot \rho^{(2)}_{S'}$.

Can we guarantee approximate decoupling when $X$ and $Z$ exhibit only approximate complementarity? Because of monogamy of correlations, this question is closely related to the question of whether the $X$ and $Z$ measurements create entanglement \cite{BuscemiNJP2009}: if $S$ is highly entangled to $M_1M_2$, then it is almost completely decoupled from some other system $S'$.  Thus, \eqref{eqn2} must imply a corresponding decoupling result. To prove this, we consider the relative entropy distance $D(\sigma\|\tau):=\Tr(\sigma\log\sigma) - \Tr(\sigma\log\tau)$~\footnote{While the relative entropy operationally measures distinguishability \cite{ohya2004quantum}, it is, mathematically speaking, not a distance; for example, it is is not symmetric in $\sigma$ and $\tau$. Nonetheless, Pinsker's inequality relates it to the trace distance: $D(\sigma\|\tau)\geq  \|\sigma-\tau\|_1^2/(2\ln2)$, $\|X\|_1 = \Tr\sqrt{X^\dagger X}$.}.
We find the following.
\begin{corollary}
\label{cor2}
For any initial $\rho^{(0)}_{SS'}$, at time $t_2$
\begin{equation}
\label{eqn5}
D(\rho^{(2)}_{SS'} || \id /\ddd  \ot \rho^{(2)}_{S'}) \leq \log (\ddd \cdot c).
\end{equation}
\end{corollary}
\begin{proof}
The state $\rho^{(2)}_{SM_1M_2}$ falls into a class of states \cite{ColesCollapsePRA2012, CorDeOFan11} for which the distillable entanglement satisfies $E(X,Z) = - H(S | M_1M_2)_{\rho^{(2)}}$. Moreover, $H(S | M_1M_2)_{\rho^{(2)}}+ H(S|S')_{\rho^{(2)}}\geq 0$ because of strong subadditivity of entropy~\cite{NieChu00}. Finally, note that $\log \ddd  - H(S|S')_{\rho^{(2)}}$ is the relative entropy on the l.h.s.\ of \eqref{eqn5}.
\end{proof}
If $X$ and $Z$ are complementary, $c=1/\ddd $ and Corollary~\ref{cor2} implies $\rho^{(2)}_{\SS'}= \id /\ddd  \ot \rho^{(2)}_{S'}$. More generally, \eqref{eqn5} shows that $S$ and $S'$ are \textit{almost} decoupled if $X$ and $Z$ are \textit{almost} complementary.

\bigskip

\emph{Coherent teleportation.---}When $X$ and $Z$ are MUBs, Proposition~\ref{prp1} says that there exists a local unitary on $M_1M_2$ that recovers the input state $\rho^{(0)}_S$. As we decrease the complementarity between $X$ and $Z$, the channel $\EC \colo S(t_0)\to S(t_2)$ goes from the completely depolarizing channel to the dephasing channel (in the limit $X=Z$), while the complementary channel $\EC^c \colo  S(t_0)\to M_1M_2(t_2)$ goes from a perfect quantum channel to a dephasing channel. One can therefore consider the quantum capacity of $\EC^c$, i.e., the optimal rate at which $\EC^c$ allows for the reliable transmission of quantum information~\cite{wilde2011classical}, as a measure of the complementarity of $X$ and $Z$. We make these ideas quantitative in the following corollary. 
\begin{corollary}
\label{cor3}
The quantum capacity $Q(\EC^c)$ of the channel $\EC^c$ satisfies $Q(\EC^c) \geq \log (1/c)$.
Furthermore, there exists a recovery map $\RC$ such that the entanglement fidelity $F_e(\RC \circ \EC^c):=\Tr\Big([\Phi]_{SS'}(\RC \circ \EC^c)_{S}([\Phi]_{SS'})\Big)$ is lower-bounded by $F_e(\RC \circ \EC^c) \geq 1/(\ddd \cdot c)$.
\end{corollary}
\begin{proof}
Suppose $\rho^{(0)}_S=\id/\ddd =\Tr_{S'} [\Phi]_{SS'}$; then from \eqref{eqn2},
\begin{align}
\log(1/c)&\leq E(X,Z)\notag \\
&=-H(S|M_1M_2)_{\rho^{(2)}} \notag \\
&=H(\rho^{(2)}_{M_1M_2})-H(\rho^{(2)}_S),
\end{align}
where the second equality follows from $H(\rho^{(2)}_{SM_1M_2})=H(\rho^{(0)}_{S})=H(\rho^{(2)}_{S})$. The third line is a lower bound on the quantum capacity of the channel $\EC^c$~\cite{wilde2011classical}.

The proof of the second claim follows from the operational meaning of the conditional min-entropy~\cite{KonRenSch09}
$H_{\min}(A|B)_{\sg} = - \log [\dim(\HC_A) \max_{\RC}\bra{\Phi}(\IC \ot \RC)(\sg_{AB})\ket{\Phi}]$, where the $\max$ is over all completely-positive trace-preserving maps $\RC$, which gives $\max_{\RC} F_e(\RC \circ \EC^c) = (1/\ddd ) 2^{-H_{\min}(S'|M_1M_2)_{\rho^{(2)}}}$, where $S'$ purifies $\rho^{(0)}_S $. Finally note that $-H_{\min}(S'|M_1M_2)_{\rho^{(2)}} \geq -H(S'|M_1M_2)_{\rho^{(2)}} = E(X,Z)$.
\end{proof}

Corollary~\ref{cor3} allows us to say that we can \textit{approximately} teleport the state $\rho^{(0)}_S$ when $X$ and $Z$ are \textit{almost} MUBs. Conceptually, Corollary~\ref{cor3} follows from \eqref{eqn2} since the latter says that $S$ becomes highly entangled to $M_1M_2$, which implies that $\rho^{(2)}_S$ must be close to the maximally mixed state regardless of the input $\rho^{(0)}_S$, which implies that $\EC$ is a bad channel and hence the complementary channel $\EC^c$ must be good~\cite{continuitystinespring}. 

\bigskip

\emph{Initially mixed devices.---}In Fig.~\ref{fgr1}, we assumed the initial states of the measurement devices were pure, $\rho^{(0)}_{M_1}= \dya{0}$ and $\rho^{(0)}_{M_2}= \dya{0}$. We now focus on the effects of mixing. While we still assume that the system-device interaction takes place on a time scale on which coherence is preserved, it is natural to restrict our attention to the case where the device's initial state is diagonal in the basis---which we have taken as the standard basis---in which the measurement result is ``recorded'': off-diagonal elements in this basis typically correspond to macroscopic superpositions and are rapidly decohered~\cite{ZurekReview}. So we write $\rho^{(0)}_{M_1} = \sum_j \al_j \dya{j}$ and $\rho^{(0)}_{M_2} = \sum_j \bt_j \dya{j}$, with $\{\alpha_j\}$ and $\{\beta_j\}$ normalized probability distributions.

For a single measurement, the effect of mixing is to reduce the ability of the device to ``accept" information \cite{VedralPRL2003}. Thus, one expects mixing to adversely affect the creation of entanglement in our setup. However, as proven in the Appendix \footnote{See the Appendix.}, we find that limited mixing only partially hinders entanglement creation. We have the following simple bound that generalizes Eq.~\eqref{eqn2} to the case of mixed devices
\begin{equation}
\label{eqn6}
E(X,Z) \geq \log (1/c) - [H(\rho^{(0)}_{M_1}) + H(\rho^{(0)}_{M_2})].
\end{equation}

For decoupling, \eqref{eqn5} will of course still hold in the case of initially mixed devices, since $\rho^{(2)}_{SS'}$ is the same regardless of whether $\rho^{(0)}_{M_1}$ and $\rho^{(0)}_{M_2}$ are mixed. For coherent teleportation, Corollary~\ref{cor3} generalizes in a simple way \footnotemark[\value{footnote}]; for example, we find
\begin{equation}
\label{eqn7}
Q(\EC^c) \geq \log (1/c) - [H(\rho^{(0)}_{M_1}) + H(\rho^{(0)}_{M_2})].
\end{equation}

\bigskip

\emph{More than two measurements.---}Our main result can be generalized in a different way. Instead of two measurements, we may consider $n\geq2$ measurements. Suppose then, that system $S$ interacts sequentially with $n$ measurement devices, each initialized in $\ket{0}$. Time $t_m$ corresponds to the time immediately after the $m$-th measurement device $M_m$, which measures observable $X^m$ of $S$, has interacted with $S$. We are interested in the entanglement at time $t_n$ between $S$ and the measurement devices $M_1\ldots M_n$, denoted $E(X^1,\ldots,X^n)$. One could also consider the entanglement at some prior time $t_m < t_n$; however, this will always be smaller than the entanglement at time $t_n$, because
\begin{equation}
\label{eqn8}
E(X^1,\ldots,X^n) \geq E(X^1,\ldots,X^{n-1}).
\end{equation}
The proof of \eqref{eqn8} notes that each measurement can be thought of as a random-unitary channel acting on $S$, where the information about which unitary is applied is stored in the measurement device. Consider the LOCC operation that extracts this information from $M_n$ and then communicates the result to $S$, allowing the local unitary on $S$ to be undone~\cite{quantumlostfound}. Thus, for every outcome this will restore the state on $SM_1\ldots M_{n-1}$ to the state at time $t_{n-1}$~\cite{PianiAdessoPRA.85.040301}. Since $E$ is non-increasing under LOCC~\cite{HHHH09}, the desired result follows.

The following bound generalizes \eqref{eqn2} to the case $n\geq 2$:
\begin{equation}
\label{eqn9}
E(X^1,\dots,X^n) \geq \max_{m< n} \log \frac{1}{c_{m,m+1}},
\end{equation}
where $c_{m,m+1} := \max_{j,k} |\ip{X^m_j}{X^{m+1}_k}|^2 $. The proof of \eqref{eqn9} is essentially the same as that of \eqref{eqn2} and is provided in \footnotemark[\value{footnote}]. Eq.~\eqref{eqn9} implies that if two MUBs are measured one after the other at any point in the sequence of measurements, then the system will become maximally entangled with the measurement devices, and  any further measurements will not generate any more entanglement.  

By the same argument in Corollary~\ref{cor2}, the analogous decoupling result follows:
\begin{equation}
\label{eqn10}
D(\rho^{(n)}_{SS'} || \id /\ddd  \ot \rho^{(n)}_{S'}) \leq \min_{m< n} \log (\ddd \cdot c_{m,m+1}),
\end{equation}
where $\rho^{(n)}_{SS'}$ is the state at time $t_n$. Likewise by the same argument in Corollary~\ref{cor3}, the analogous coherent teleportation result follows:
\begin{equation}
\label{eqn11}
Q(\EC^c) \geq \max_{m< n} \log \frac{1}{c_{m,m+1}},
\end{equation}
where $\EC^c$ is the channel from $S$ at $t_0$ to $M_1\ldots M_n$ at $t_n$, and the analogous generalization for $F_e$ also holds.

\bigskip

\emph{Conclusions.---}We proposed that a signature and a quantification of complementarity of two observables is  given by the entanglement generated when the two observables are sequentially measured on the same system by means of a coherent interaction with corresponding measurement devices. We also noted how this approach to complementarity is intimately related to the {information}-processing primitives of decoupling and coherent teleportation. 

The importance of complementarity in quantum information processing has been explored previously, e.g., by Renes and collaborators (see \cite{RenesHabThesis} and references therein). Such works typically focus on the transmission of information in complementary bases, which turns out to be sufficient to ensure transmission of quantum information. However, the physical scenario of sequential coherent complementary measurements is not obviously connected to mathematical theorems \cite{ChristWinterIEEE2005, RenesBoileau, Renes2010arXiv1003.1150R, ColesEtAlPRA2011} regarding the knowledge or transmission of complementary information, particularly in the case of partial complementarity.

The fact that, in our scheme, the complementarity of two observables measures their power to process quantum information suggests to search for further ``uncertainty'' (or ``certainty'') relations for other information-processing tasks or quantum computing algorithms. Ref.~\cite{BranHoro2010arXiv1010.3654B} already made some progress along these lines, and we expect that our work will stimulate further results in the same perspective.

\bigskip

\emph{Note added.---}One of us (PJC) coauthored also \cite{BertaColesWehner2013}. There, pre-existing entanglement is connected to the uncertainty of measurements on distinct but identically-prepared systems. Such work is not closely related to the present one, since we consider dynamically-created entanglement during sequential measurements.

\bigskip

\emph{Acknowledgments.---}We thank Stephanie Wehner and Takafumi Nakano for helpful discussions. PJC is funded by the Ministry of Education (MOE) and National Research Foundation Singapore, as well as MOE Tier 3 Grant ``Random numbers from quantum processes" (MOE2012-T3-1-009). MP acknowledges support from NSERC, CIFAR, and Ontario Centres of Excellence.

\bibliography{ECR}

\begin{thebibliography}{44}
\expandafter\ifx\csname natexlab\endcsname\relax\def\natexlab#1{#1}\fi
\expandafter\ifx\csname bibnamefont\endcsname\relax
  \def\bibnamefont#1{#1}\fi
\expandafter\ifx\csname bibfnamefont\endcsname\relax
  \def\bibfnamefont#1{#1}\fi
\expandafter\ifx\csname citenamefont\endcsname\relax
  \def\citenamefont#1{#1}\fi
\expandafter\ifx\csname url\endcsname\relax
  \def\url#1{\texttt{#1}}\fi
\expandafter\ifx\csname urlprefix\endcsname\relax\def\urlprefix{URL }\fi
\providecommand{\bibinfo}[2]{#2}
\providecommand{\eprint}[2][]{\url{#2}}

\bibitem[{\citenamefont{Price and Chissick}(1977)}]{price1977uncertainty}
\bibinfo{author}{\bibfnamefont{W.~C.} \bibnamefont{Price}} \bibnamefont{and}
  \bibinfo{author}{\bibfnamefont{S.~S.} \bibnamefont{Chissick}},
  \emph{\bibinfo{title}{The uncertainty principle and foundations of quantum
  mechanics}} (\bibinfo{publisher}{John Wiley \& Sons}, \bibinfo{year}{1977}).

\bibitem[{\citenamefont{Ozawa}(2003)}]{ozawa_pra_2003}
\bibinfo{author}{\bibfnamefont{M.}~\bibnamefont{Ozawa}},
  \bibinfo{journal}{Phys. Rev. A} \textbf{\bibinfo{volume}{67}},
  \bibinfo{pages}{042105} (\bibinfo{year}{2003}).

\bibitem[{\citenamefont{Heisenberg}(1927)}]{Heisenberg}
\bibinfo{author}{\bibfnamefont{W.}~\bibnamefont{Heisenberg}},
  \bibinfo{journal}{Zeitschrift f\"ur Physik} \textbf{\bibinfo{volume}{43}},
  \bibinfo{pages}{172} (\bibinfo{year}{1927}).

\bibitem[{\citenamefont{Kennard}(1927)}]{kennard1927quantum}
\bibinfo{author}{\bibfnamefont{E.}~\bibnamefont{Kennard}}, \bibinfo{journal}{Z.
  Phys} \textbf{\bibinfo{volume}{44}}, \bibinfo{pages}{326}
  (\bibinfo{year}{1927}).

\bibitem[{\citenamefont{Robertson}(1929)}]{Robertson}
\bibinfo{author}{\bibfnamefont{H.~P.} \bibnamefont{Robertson}},
  \bibinfo{journal}{Phys. Rev.} \textbf{\bibinfo{volume}{34}},
  \bibinfo{pages}{163} (\bibinfo{year}{1929}).

\bibitem[{\citenamefont{{Wehner} and {Winter}}(2010)}]{EURreview1}
\bibinfo{author}{\bibfnamefont{S.}~\bibnamefont{{Wehner}}} \bibnamefont{and}
  \bibinfo{author}{\bibfnamefont{A.}~\bibnamefont{{Winter}}},
  \bibinfo{journal}{New J. Phys.} \textbf{\bibinfo{volume}{12}},
  \bibinfo{pages}{025009} (\bibinfo{year}{2010}).

\bibitem[{\citenamefont{Maassen and Uffink}(1988)}]{MaassenUffink}
\bibinfo{author}{\bibfnamefont{H.}~\bibnamefont{Maassen}} \bibnamefont{and}
  \bibinfo{author}{\bibfnamefont{J.~B.~M.} \bibnamefont{Uffink}},
  \bibinfo{journal}{Phys. Rev. Lett.} \textbf{\bibinfo{volume}{60}},
  \bibinfo{pages}{1103} (\bibinfo{year}{1988}).

\bibitem[{\citenamefont{Horodecki et~al.}(2009)\citenamefont{Horodecki,
  Horodecki, Horodecki, and Horodecki}}]{HHHH09}
\bibinfo{author}{\bibfnamefont{R.}~\bibnamefont{Horodecki}},
  \bibinfo{author}{\bibfnamefont{P.}~\bibnamefont{Horodecki}},
  \bibinfo{author}{\bibfnamefont{M.}~\bibnamefont{Horodecki}},
  \bibnamefont{and}
  \bibinfo{author}{\bibfnamefont{K.}~\bibnamefont{Horodecki}},
  \bibinfo{journal}{Rev. Mod. Phys.} \textbf{\bibinfo{volume}{81}},
  \bibinfo{pages}{865} (\bibinfo{year}{2009}).

\bibitem[{\citenamefont{Coles}(2012{\natexlab{a}})}]{ColesCollapsePRA2012}
\bibinfo{author}{\bibfnamefont{P.~J.} \bibnamefont{Coles}},
  \bibinfo{journal}{Phys. Rev. A} \textbf{\bibinfo{volume}{86}},
  \bibinfo{pages}{062334} (\bibinfo{year}{2012}{\natexlab{a}}).

\bibitem[{\citenamefont{Streltsov et~al.}(2011)\citenamefont{Streltsov,
  Kampermann, and Bru\ss{}}}]{StrKamBru11}
\bibinfo{author}{\bibfnamefont{A.}~\bibnamefont{Streltsov}},
  \bibinfo{author}{\bibfnamefont{H.}~\bibnamefont{Kampermann}},
  \bibnamefont{and} \bibinfo{author}{\bibfnamefont{D.}~\bibnamefont{Bru\ss{}}},
  \bibinfo{journal}{Phys. Rev. Lett.} \textbf{\bibinfo{volume}{106}},
  \bibinfo{pages}{160401} (\bibinfo{year}{2011}).

\bibitem[{\citenamefont{Piani et~al.}(2011)\citenamefont{Piani, Gharibian,
  Adesso, Calsamiglia, Horodecki, and Winter}}]{PianiEtAl11}
\bibinfo{author}{\bibfnamefont{M.}~\bibnamefont{Piani}},
  \bibinfo{author}{\bibfnamefont{S.}~\bibnamefont{Gharibian}},
  \bibinfo{author}{\bibfnamefont{G.}~\bibnamefont{Adesso}},
  \bibinfo{author}{\bibfnamefont{J.}~\bibnamefont{Calsamiglia}},
  \bibinfo{author}{\bibfnamefont{P.}~\bibnamefont{Horodecki}},
  \bibnamefont{and} \bibinfo{author}{\bibfnamefont{A.}~\bibnamefont{Winter}},
  \bibinfo{journal}{Phys. Rev. Lett.} \textbf{\bibinfo{volume}{106}},
  \bibinfo{pages}{220403} (\bibinfo{year}{2011}).

\bibitem[{\citenamefont{{Gharibian} et~al.}(2011)\citenamefont{{Gharibian},
  {Piani}, {Adesso}, {Calsamiglia}, and {Horodecki}}}]{GharEtAl2011}
\bibinfo{author}{\bibfnamefont{S.}~\bibnamefont{{Gharibian}}},
  \bibinfo{author}{\bibfnamefont{M.}~\bibnamefont{{Piani}}},
  \bibinfo{author}{\bibfnamefont{G.}~\bibnamefont{{Adesso}}},
  \bibinfo{author}{\bibfnamefont{J.}~\bibnamefont{{Calsamiglia}}},
  \bibnamefont{and}
  \bibinfo{author}{\bibfnamefont{P.}~\bibnamefont{{Horodecki}}},
  \bibinfo{journal}{International Journal of Quantum Information}
  \textbf{\bibinfo{volume}{9}}, \bibinfo{pages}{1701} (\bibinfo{year}{2011}),
  \bibinfo{note}{eprint arXiv:1105.3419 [quant-ph]}.

\bibitem[{\citenamefont{Piani and Adesso}(2012)}]{PianiAdessoPRA.85.040301}
\bibinfo{author}{\bibfnamefont{M.}~\bibnamefont{Piani}} \bibnamefont{and}
  \bibinfo{author}{\bibfnamefont{G.}~\bibnamefont{Adesso}},
  \bibinfo{journal}{Phys. Rev. A} \textbf{\bibinfo{volume}{85}},
  \bibinfo{pages}{040301} (\bibinfo{year}{2012}).

\bibitem[{\citenamefont{Schumacher and Westmoreland}(2002)}]{SchuWestErrCorr}
\bibinfo{author}{\bibfnamefont{B.}~\bibnamefont{Schumacher}} \bibnamefont{and}
  \bibinfo{author}{\bibfnamefont{M.~D.} \bibnamefont{Westmoreland}},
  \bibinfo{journal}{Quantum Information Processing}
  \textbf{\bibinfo{volume}{1}}, \bibinfo{pages}{5} (\bibinfo{year}{2002}), ISSN
  \bibinfo{issn}{1570-0755}.

\bibitem[{\citenamefont{{Dupuis} et~al.}(2010)\citenamefont{{Dupuis}, {Berta},
  {Wullschleger}, and {Renner}}}]{DupEtAl2010}
\bibinfo{author}{\bibfnamefont{F.}~\bibnamefont{{Dupuis}}},
  \bibinfo{author}{\bibfnamefont{M.}~\bibnamefont{{Berta}}},
  \bibinfo{author}{\bibfnamefont{J.}~\bibnamefont{{Wullschleger}}},
  \bibnamefont{and} \bibinfo{author}{\bibfnamefont{R.}~\bibnamefont{{Renner}}},
  \bibinfo{journal}{ArXiv e-prints}  (\bibinfo{year}{2010}),
  \bibinfo{note}{1012.6044v2}.

\bibitem[{\citenamefont{{Dupuis}}(2009)}]{DupuisThesis}
\bibinfo{author}{\bibfnamefont{F.}~\bibnamefont{{Dupuis}}}, Ph.D. thesis,
  \bibinfo{school}{Universit\'e de Montr\'eal} (\bibinfo{year}{2009}),
  \urlprefix\url{http://arxiv.org/abs/1004.1641}.

\bibitem[{\citenamefont{Groisman et~al.}(2005)\citenamefont{Groisman, Popescu,
  and Winter}}]{GroPopWinPRA72.032317}
\bibinfo{author}{\bibfnamefont{B.}~\bibnamefont{Groisman}},
  \bibinfo{author}{\bibfnamefont{S.}~\bibnamefont{Popescu}}, \bibnamefont{and}
  \bibinfo{author}{\bibfnamefont{A.}~\bibnamefont{Winter}},
  \bibinfo{journal}{Phys. Rev. A} \textbf{\bibinfo{volume}{72}},
  \bibinfo{pages}{032317} (\bibinfo{year}{2005}).

\bibitem[{\citenamefont{Buscemi}(2009)}]{BuscemiNJP2009}
\bibinfo{author}{\bibfnamefont{F.}~\bibnamefont{Buscemi}},
  \bibinfo{journal}{New Journal of Physics} \textbf{\bibinfo{volume}{11}},
  \bibinfo{pages}{123002} (\bibinfo{year}{2009}).

\bibitem[{\citenamefont{Brassard et~al.}(1998)\citenamefont{Brassard,
  Braunstein, and Cleve}}]{Brassard1998}
\bibinfo{author}{\bibfnamefont{G.}~\bibnamefont{Brassard}},
  \bibinfo{author}{\bibfnamefont{S.~L.} \bibnamefont{Braunstein}},
  \bibnamefont{and} \bibinfo{author}{\bibfnamefont{R.}~\bibnamefont{Cleve}},
  \bibinfo{journal}{Physica D: Nonlinear Phenomena}
  \textbf{\bibinfo{volume}{120}}, \bibinfo{pages}{43 } (\bibinfo{year}{1998}),
  ISSN \bibinfo{issn}{0167-2789}, \bibinfo{note}{proceedings of the Fourth
  Workshop on Physics and Consumption}.

\bibitem[{\citenamefont{Harrow}(2004)}]{HarrowPRL2004}
\bibinfo{author}{\bibfnamefont{A.}~\bibnamefont{Harrow}},
  \bibinfo{journal}{Phys. Rev. Lett.} \textbf{\bibinfo{volume}{92}},
  \bibinfo{pages}{097902} (\bibinfo{year}{2004}).

\bibitem[{\citenamefont{Zurek}(2003)}]{ZurekReview}
\bibinfo{author}{\bibfnamefont{W.~H.} \bibnamefont{Zurek}},
  \bibinfo{journal}{Rev. Mod. Phys.} \textbf{\bibinfo{volume}{75}},
  \bibinfo{pages}{715} (\bibinfo{year}{2003}).

\bibitem[{\citenamefont{Seevinck}(2010)}]{seevinck2010}
\bibinfo{author}{\bibfnamefont{M.}~\bibnamefont{Seevinck}},
  \bibinfo{journal}{Quantum Information Processing}
  \textbf{\bibinfo{volume}{9}}, \bibinfo{pages}{273} (\bibinfo{year}{2010}),
  ISSN \bibinfo{issn}{1570-0755}.

\bibitem[{\citenamefont{{Berta} et~al.}(2010)\citenamefont{{Berta},
  {Christandl}, {Colbeck}, {Renes}, and {Renner}}}]{BertaEtAl}
\bibinfo{author}{\bibfnamefont{M.}~\bibnamefont{{Berta}}},
  \bibinfo{author}{\bibfnamefont{M.}~\bibnamefont{{Christandl}}},
  \bibinfo{author}{\bibfnamefont{R.}~\bibnamefont{{Colbeck}}},
  \bibinfo{author}{\bibfnamefont{J.~M.} \bibnamefont{{Renes}}},
  \bibnamefont{and} \bibinfo{author}{\bibfnamefont{R.}~\bibnamefont{{Renner}}},
  \bibinfo{journal}{Nature Physics} \textbf{\bibinfo{volume}{6}},
  \bibinfo{pages}{659} (\bibinfo{year}{2010}).

\bibitem[{\citenamefont{Coles}(2012{\natexlab{b}})}]{ColesDecDisc2012}
\bibinfo{author}{\bibfnamefont{P.~J.} \bibnamefont{Coles}},
  \bibinfo{journal}{Phys. Rev. A} \textbf{\bibinfo{volume}{85}},
  \bibinfo{pages}{042103} (\bibinfo{year}{2012}{\natexlab{b}}).

\bibitem[{\citenamefont{Horodecki et~al.}(2005)\citenamefont{Horodecki, Winter
  et~al.}}]{horodecki2005partial}
\bibinfo{author}{\bibfnamefont{J.~O.} \bibnamefont{Horodecki}},
  \bibinfo{author}{\bibfnamefont{A.}~\bibnamefont{Winter}},
  \bibnamefont{et~al.}, \bibinfo{journal}{Nature}
  \textbf{\bibinfo{volume}{436}}, \bibinfo{pages}{673} (\bibinfo{year}{2005}).

\bibitem[{\citenamefont{Gisin et~al.}(2002)\citenamefont{Gisin, Ribordy,
  Tittel, and Zbinden}}]{RevModPhysCrypto}
\bibinfo{author}{\bibfnamefont{N.}~\bibnamefont{Gisin}},
  \bibinfo{author}{\bibfnamefont{G.}~\bibnamefont{Ribordy}},
  \bibinfo{author}{\bibfnamefont{W.}~\bibnamefont{Tittel}}, \bibnamefont{and}
  \bibinfo{author}{\bibfnamefont{H.}~\bibnamefont{Zbinden}},
  \bibinfo{journal}{Rev. Mod. Phys.} \textbf{\bibinfo{volume}{74}},
  \bibinfo{pages}{145} (\bibinfo{year}{2002}).

\bibitem[{\citenamefont{Cornelio et~al.}(2011)\citenamefont{Cornelio,
  de~Oliveira, and Fanchini}}]{CorDeOFan11}
\bibinfo{author}{\bibfnamefont{M.~F.} \bibnamefont{Cornelio}},
  \bibinfo{author}{\bibfnamefont{M.~C.} \bibnamefont{de~Oliveira}},
  \bibnamefont{and} \bibinfo{author}{\bibfnamefont{F.~F.}
  \bibnamefont{Fanchini}}, \bibinfo{journal}{Phys. Rev. Lett.}
  \textbf{\bibinfo{volume}{107}}, \bibinfo{pages}{020502}
  (\bibinfo{year}{2011}).

\bibitem[{\citenamefont{Nielsen and Chuang}(2000)}]{NieChu00}
\bibinfo{author}{\bibfnamefont{M.~A.} \bibnamefont{Nielsen}} \bibnamefont{and}
  \bibinfo{author}{\bibfnamefont{I.~L.} \bibnamefont{Chuang}},
  \emph{\bibinfo{title}{Quantum Computation and Quantum Information}}
  (\bibinfo{publisher}{Cambridge University Press},
  \bibinfo{address}{Cambridge}, \bibinfo{year}{2000}), \bibinfo{edition}{5th}
  ed.

\bibitem[{\citenamefont{Wilde}(2011)}]{wilde2011classical}
\bibinfo{author}{\bibfnamefont{M.~M.} \bibnamefont{Wilde}},
  \bibinfo{journal}{arXiv preprint arXiv:1106.1445}  (\bibinfo{year}{2011}).

\bibitem[{\citenamefont{Konig et~al.}(2009)\citenamefont{Konig, Renner, and
  Schaffner}}]{KonRenSch09}
\bibinfo{author}{\bibfnamefont{R.}~\bibnamefont{Konig}},
  \bibinfo{author}{\bibfnamefont{R.}~\bibnamefont{Renner}}, \bibnamefont{and}
  \bibinfo{author}{\bibfnamefont{C.}~\bibnamefont{Schaffner}},
  \bibinfo{journal}{IEEE Trans. Inf. Theory} \textbf{\bibinfo{volume}{55}},
  \bibinfo{pages}{4337 } (\bibinfo{year}{2009}).

\bibitem[{\citenamefont{Kretschmann et~al.}(2008)\citenamefont{Kretschmann,
  Schlingemann, and Werner}}]{continuitystinespring}
\bibinfo{author}{\bibfnamefont{D.}~\bibnamefont{Kretschmann}},
  \bibinfo{author}{\bibfnamefont{D.}~\bibnamefont{Schlingemann}},
  \bibnamefont{and} \bibinfo{author}{\bibfnamefont{R.}~\bibnamefont{Werner}},
  \bibinfo{journal}{Information Theory, IEEE Transactions on}
  \textbf{\bibinfo{volume}{54}}, \bibinfo{pages}{1708} (\bibinfo{year}{2008}),
  ISSN \bibinfo{issn}{0018-9448}.

\bibitem[{\citenamefont{Vedral}(2003)}]{VedralPRL2003}
\bibinfo{author}{\bibfnamefont{V.}~\bibnamefont{Vedral}},
  \bibinfo{journal}{Phys. Rev. Lett.} \textbf{\bibinfo{volume}{90}},
  \bibinfo{pages}{050401} (\bibinfo{year}{2003}).

\bibitem[{\citenamefont{Gregoratti and Werner}(2003)}]{quantumlostfound}
\bibinfo{author}{\bibfnamefont{M.}~\bibnamefont{Gregoratti}} \bibnamefont{and}
  \bibinfo{author}{\bibfnamefont{R.~F.} \bibnamefont{Werner}},
  \bibinfo{journal}{Journal of Modern Optics} \textbf{\bibinfo{volume}{50}},
  \bibinfo{pages}{915} (\bibinfo{year}{2003}).

\bibitem[{\citenamefont{{Renes}}(2012)}]{RenesHabThesis}
\bibinfo{author}{\bibfnamefont{J.~M.} \bibnamefont{{Renes}}},
  \bibinfo{journal}{ArXiv e-prints}  (\bibinfo{year}{2012}),
  \bibinfo{note}{1212.2379}.

\bibitem[{\citenamefont{Christandl and Winter}(2005)}]{ChristWinterIEEE2005}
\bibinfo{author}{\bibfnamefont{M.}~\bibnamefont{Christandl}} \bibnamefont{and}
  \bibinfo{author}{\bibfnamefont{A.}~\bibnamefont{Winter}},
  \bibinfo{journal}{IEEE Trans. Inf. Theory} \textbf{\bibinfo{volume}{51}},
  \bibinfo{pages}{3159} (\bibinfo{year}{2005}).

\bibitem[{\citenamefont{Renes and Boileau}(2009)}]{RenesBoileau}
\bibinfo{author}{\bibfnamefont{J.~M.} \bibnamefont{Renes}} \bibnamefont{and}
  \bibinfo{author}{\bibfnamefont{J.-C.} \bibnamefont{Boileau}},
  \bibinfo{journal}{Phys. Rev. Lett.} \textbf{\bibinfo{volume}{103}},
  \bibinfo{pages}{020402} (\bibinfo{year}{2009}).

\bibitem[{\citenamefont{{Renes}}(2010)}]{Renes2010arXiv1003.1150R}
\bibinfo{author}{\bibfnamefont{J.~M.} \bibnamefont{{Renes}}},
  \bibinfo{journal}{ArXiv e-prints}  (\bibinfo{year}{2010}),
  \eprint{1003.1150}.

\bibitem[{\citenamefont{Coles et~al.}(2011)\citenamefont{Coles, Yu, Gheorghiu,
  and Griffiths}}]{ColesEtAlPRA2011}
\bibinfo{author}{\bibfnamefont{P.~J.} \bibnamefont{Coles}},
  \bibinfo{author}{\bibfnamefont{L.}~\bibnamefont{Yu}},
  \bibinfo{author}{\bibfnamefont{V.}~\bibnamefont{Gheorghiu}},
  \bibnamefont{and} \bibinfo{author}{\bibfnamefont{R.~B.}
  \bibnamefont{Griffiths}}, \bibinfo{journal}{Phys. Rev. A}
  \textbf{\bibinfo{volume}{83}}, \bibinfo{pages}{062338}
  (\bibinfo{year}{2011}).

\bibitem[{\citenamefont{{Brandao} and
  {Horodecki}}(2013)}]{BranHoro2010arXiv1010.3654B}
\bibinfo{author}{\bibfnamefont{F.~G.~S.~L.} \bibnamefont{{Brandao}}}
  \bibnamefont{and}
  \bibinfo{author}{\bibfnamefont{M.}~\bibnamefont{{Horodecki}}},
  \bibinfo{journal}{Q. Inf. Comp.} \textbf{\bibinfo{volume}{13}},
  \bibinfo{pages}{0901} (\bibinfo{year}{2013}).

\bibitem[{\citenamefont{{Berta} et~al.}(2013)\citenamefont{{Berta}, {Coles},
  and {Wehner}}}]{BertaColesWehner2013}
\bibinfo{author}{\bibfnamefont{M.}~\bibnamefont{{Berta}}},
  \bibinfo{author}{\bibfnamefont{P.~J.} \bibnamefont{{Coles}}},
  \bibnamefont{and} \bibinfo{author}{\bibfnamefont{S.}~\bibnamefont{{Wehner}}},
  \bibinfo{journal}{ArXiv e-prints}  (\bibinfo{year}{2013}),
  \eprint{1302.5902}.

\bibitem[{\citenamefont{Ohya and Petz}(2004)}]{ohya2004quantum}
\bibinfo{author}{\bibfnamefont{M.}~\bibnamefont{Ohya}} \bibnamefont{and}
  \bibinfo{author}{\bibfnamefont{D.}~\bibnamefont{Petz}},
  \emph{\bibinfo{title}{Quantum entropy and its use}}
  (\bibinfo{publisher}{Springer Verlag}, \bibinfo{year}{2004}).

\bibitem[{\citenamefont{{Tomamichel} and {Renner}}(2011)}]{TomRen2010}
\bibinfo{author}{\bibfnamefont{M.}~\bibnamefont{{Tomamichel}}}
  \bibnamefont{and} \bibinfo{author}{\bibfnamefont{R.}~\bibnamefont{{Renner}}},
  \bibinfo{journal}{Phys. Rev. Lett.} \textbf{\bibinfo{volume}{106}},
  \bibinfo{pages}{110506} (\bibinfo{year}{2011}).

\bibitem[{\citenamefont{Devetak and Winter}(2005)}]{DevWin05}
\bibinfo{author}{\bibfnamefont{I.}~\bibnamefont{Devetak}} \bibnamefont{and}
  \bibinfo{author}{\bibfnamefont{A.}~\bibnamefont{Winter}},
  \bibinfo{journal}{Proc. R. Soc. A} \textbf{\bibinfo{volume}{461}},
  \bibinfo{pages}{207} (\bibinfo{year}{2005}).

\bibitem[{\citenamefont{Tomamichel et~al.}(2009)\citenamefont{Tomamichel,
  Colbeck, and Renner}}]{TomColRen09}
\bibinfo{author}{\bibfnamefont{M.}~\bibnamefont{Tomamichel}},
  \bibinfo{author}{\bibfnamefont{R.}~\bibnamefont{Colbeck}}, \bibnamefont{and}
  \bibinfo{author}{\bibfnamefont{R.}~\bibnamefont{Renner}},
  \bibinfo{journal}{IEEE Trans. Inf. Theory} \textbf{\bibinfo{volume}{55}},
  \bibinfo{pages}{5840 } (\bibinfo{year}{2009}).

\end{thebibliography}

\pagebreak

\appendix

\section{Various measures of entanglement}\label{app1}

Here we define various measures of entanglement for which our main result holds. That is, the bound:
\begin{equation}
\label{eqn12}
E(X,Z) \geq \log(1/c)
\end{equation}
was stated in the main text where $E$ was assumed to be the distillable entanglement, but we discuss here that several other measures of entanglement also obey this bound.

Consider the following measures of entanglement for some bipartite state $\rho_{AB}$~\cite{HHHH09}:\\

(1) $E_D$, distillable entanglement: the optimal rate to distill EPR pairs using LOCC in the asymptotic limit of infinitely many copies of $\rho_{AB}$.\\

(2) $K$, distillable secret key: the optimal rate to distill bits of secret key using LOCC in the asymptotic limit of infinitely many copies of $\rho_{AB}$.\\

(3) $E_{F}$, Entanglement of formation: $E_F(\rho_{AB}):= \min_{\{\ket{\phi_j}\}} \sum_j p_j H[\Tr_B(\dya{\phi_j})]$, where the minimization is over all convex decompositions of $\rho_{AB} = \sum_j p_j \dya{\phi_j}$.\\

(4) $E_{C}$, Entanglement cost: the regularization of $E_F$, $E_C(\rho_{AB}) = \lim_{N\to \infty} (1/N) E_F(\rho_{AB}^{\ot N})$.\\

(5) $E_{\text{sq}}$, squashed entanglement: $E_{\text{sq}}(\rho_{AB}) = (1/2)\min_C I(A\colo B | C)$, where $I(A\colo B | C)$ is the conditional mutual information, and the minimization is over all extensions $\rho_{ABC}$ of $\rho_{AB}$. \\

(6) $E_{R}$, relative entropy of entanglement: $E_{R}(\rho_{AB}) = \min_{\sg_{AB}\in \text{Sep}}D(\rho_{AB}|| \sg_{AB})$, where the minimization is over all separable states $\sg_{AB}$. \\

(7) $E_{R,\infty}$, regularized relative entropy of entanglement: $E_{R,\infty}(\rho_{AB}) = \lim_{N\to \infty} (1/N) E_R(\rho_{AB}^{\ot N}) $.\\

(8) $E_{\max}$, max relative entropy of entanglement: $E_{\max}(\rho_{AB}) = \min_{\sg_{AB}\in \text{Sep}}D_{\max}(\rho_{AB}|| \sg_{AB})$, where the minimization is over all separable states $\sg_{AB}$, and where $D_{\max}(\rho || \sg):=\log\min\{\lm : \rho\leq \lm \sg \}$. \\

(9) $E_{\fid}$, fidelity relative entropy of entanglement: $E_{\fid}(\rho_{AB}) = \min_{\sg_{AB}\in \text{Sep}}D_{\fid}(\rho_{AB}|| \sg_{AB})$, where the minimization is over all separable states $\sg_{AB}$, and where $D_{\fid}(\rho || \sg):= - 2  \log \Tr [(\sqrt{\rho} \sg \sqrt{\rho})^{1/2}]$. \\

\begin{proposition}
Equation~\eqref{eqn12} holds for all of the entanglement measures in the above list.
\end{proposition}
\begin{proof}
In the main text, we proved this bound for $E_D$. Now note that $E_D$ is a lower bound on each of the measures $K$, $E_F$, $E_C$, $E_{\text{sq}}$, $E_R$, $E_{R,\infty}$, and $E_{\max}$, hence \eqref{eqn12} must also hold for each of these measures. For $E_{\fid}$ we replicate our proof in the main text based on the uncertainty principle with quantum memory, except this time we use the uncertainty relation for the min and max entropies from Ref.~\cite{TomRen2010}. Applying this uncertainty relation at time $t_1$ in Fig.~1 (from the main text) gives
$$H_{\max}(X|M_1M_2)_{\rho^{(1)}}+ H_{\min} (Z|S')_{\rho^{(1)}} \geq \log(1/c)$$
where $S'$ purifies $\rho^{(0)}_S$. The proof follows by noting that $H_{\max}(X|M_1M_2)_{\rho^{(1)}} =0$ since $M_1$ already measured $X$, and $H_{\min} (Z|S')_{\rho^{(1)}}$ is equal to the entanglement at time $t_2$ between $S$ and $M_1M_2$ as quantified by $E_{\fid}$ \cite{ColesDecDisc2012, ColesCollapsePRA2012}.
\end{proof}

\section{Initially mixed devices}

Here we generalize our results to the case where the measurement devices are initially in mixed states. As noted in the main text, we assume the devices' initial states are diagonal in the standard basis, i.e., the devices have been decohered in their pointer bases. Our extension to mixed devices is aided by the following lemma.

\begin{lemma}
\label{lem23496}
Let $\rho_{AB}= \sum_j p_j \rho_{AB,j}$ be a mixture of bipartite states $\{\rho_{AB,j}\}$ according to probability distribution $\{p_j\}$. Then
\begin{equation}
\label{eqn13}
-H(A|B)_{\rho} \geq \sum_j p_j [-H(A|B)_{\rho_j}] - H(\{p_j\})
\end{equation}
where $H(A|B)_{\rho_j}$ denotes the conditional entropy of $\rho_{AB,j}$.
\end{lemma}
\begin{proof}
This is a straightforward entropic inequality, resulting from combining concavity of the entropy $H(\rho_B) \geq \sum_j p_j H(\rho_{B,j})$ with the inequality $ H(\{p_j\}) + \sum_j p_j H(\rho_{AB,j}) \geq H(\rho_{AB})$ \cite{NieChu00}.
\end{proof}

With this lemma, we obtain the following corollary of our main result, which extends this result to initially mixed devices.

\begin{corollary}
\label{thmMixed2505}
Consider the paradigm discussed in the main text, where the observables $X$ and $Z$ are sequentially measured, as shown in Fig.~1.  Let $E(X,Z)$ denote the distillable entanglement at time $t_2$ between $S$ and $M_1M_2$. Let $\rho^{(0)}_{M_1} = \sum_j \al_j \dya{j}$ and $\rho^{(0)}_{M_2} = \sum_j \bt_j \dya{j}$ be possibly mixed states. Then,
\begin{equation}
\label{eqn14}
E(X,Z) \geq \log (1/c) - [H(\rho^{(0)}_{M_1}) + H(\rho^{(0)}_{M_2})].
\end{equation}
\end{corollary}
\begin{proof}
Expanding $\rho^{(0)}_{M_1}$ and $\rho^{(0)}_{M_2}$ allows us to write the state at time $t_2$ as:
\begin{equation}
\rho^{(2)}_{SM_1M_2} = \sum_{q,r} \al_q \bt_r \rho^{(2)}_{SM_1M_2,q,r}
\end{equation}
where 
\begin{align}
\rho^{(2)}_{SM_1M_2,q,r}=\sum_{j,k,l,m} &[Z_l][X_j]\rho^{(0)}_S[X_k][Z_m] \notag \\
&\ot \dyad{q+j}{q+k}\ot \dyad{r+l}{r+m}.\notag
\end{align}
Applying Lemma~\ref{lem23496} gives
\begin{align}
\label{eqn15}
-&H(S|M_1M_2)_{\rho^{(2)}} \notag \\
&\geq \sum_{q,r} \al_q \bt_r [-H(S|M_1M_2)_{\rho^{(2)}_{q,r}}] - H(\{\al_q \bt_r\})\notag \\
&\geq \sum_{q,r} \al_q \bt_r \log (1/c)- H(\{\al_q \bt_r\})\notag \\
&=  \log (1/c)- [H(\rho^{(0)}_{M_1})+H(\rho^{(0)}_{M_2})]
\end{align}
Here, the second inequality notes that the correlations across the $S$:$M_1M_2$ cut are independent of the value of $q$ and $r$, so we can set $q=r=0$ and note that $-H(S|M_1M_2)_{\rho^{(2)}_{0,0}}$ is equal to the entanglement that we lower bounded in our main result by $\log(1/c)$. The last line of \eqref{eqn15} uses the additivity of the entropy to obtain $H(\{\al_q \bt_r\}) = H(\{\al_q \})+H(\{\bt_r \})$. Finally, from Ref.~\cite{DevWin05} we have $E(X,Z)\geq - H(S|M_1M_2)_{\rho^{(2)}}$, which, combined with \eqref{eqn15}, proves the desired result.\end{proof}

Now consider the perspective of coherent teleportation. Corollary~4 generalizes nicely to the case of mixed devices as follows.  
\begin{corollary}
\label{cor127}
Let $\rho^{(0)}_{M_1} = \sum_j \al_j \dya{j}$ and $\rho^{(0)}_{M_2} = \sum_j \bt_j \dya{j}$ be possibly mixed states, and let $\EC^c$ be the quantum channel from $S$ at time $t_0$ to $M_1M_2$ at time $t_2$. Then:

(a) it holds
\begin{equation}
\label{eqn16}
Q(\EC^c) \geq \log (1/c) - [H(\rho^{(0)}_{M_1}) + H(\rho^{(0)}_{M_2})];
\end{equation}

(b) there exists a recovery map $\RC$ such that the entanglement fidelity of the channel $\RC \circ \EC^c$ is bounded by:
\begin{equation}
\label{eqn17}
F_e(\RC \circ \EC^c) \geq  \frac{1}{\ddd \cdot c} 2^{-[H(\rho^{(0)}_{M_1}) + H(\rho^{(0)}_{M_2})]}.
\end{equation}
\end{corollary}
\begin{proof}
In proving both (a) and (b), we will invoke the proof of Cor.~\ref{thmMixed2505} and we will set the initial state to $\rho^{(0)}_S = \id /\ddd $. For this input state, with $\EC^c$ and $\EC$ being complementary quantum channels, and letting $M'_1$ and $M'_2$ be systems that purify $\rho^{(0)}_{M_1}$ and $\rho^{(0)}_{M_2}$ respectively, we have
\begin{align}
\label{eqn18}
Q(\EC^c) &\geq H(\EC^c(\id /\ddd ))-H(\EC(\id /\ddd ))\notag\\
&= H(\rho^{(2)}_{M_1M_2}) - H(\rho^{(2)}_{S M'_1M'_2}) \notag \\
&= H(\rho^{(2)}_{M_1M_2}) - H(\rho^{(2)}_{S M_1M_2}) \notag\\
&= -H(S|M_1M_2)_{\rho^{(2)}}
\end{align}
In the third line, we noted that $H(\rho^{(2)}_{S M'_1M'_2}) = H(\rho^{(0)}_{S M_1M_2}) = H(\rho^{(2)}_{S M_1M_2}) $ since $\rho^{(2)}_{S M'_1M'_2} = \id/\ddd  \ot \rho^{(2)}_{M'_1} \ot \rho^{(2)}_{M'_2}$. Finally, combining \eqref{eqn18} with \eqref{eqn15} proves \eqref{eqn16}.

For \eqref{eqn17}, letting $S'$ purify $\rho^{(0)}_{S}$, we write
\begin{align}
\max_{\RC} F_e(\RC \circ \EC^c) &= (1/\ddd ) 2^{-H_{\min}(S'|M_1M_2)_{\rho^{(2)}}}\notag\\
&\geq (1/\ddd ) 2^{-H(S'|M_1M_2)_{\rho^{(2)}}} \notag \\
&= (1/\ddd ) 2^{-H(S|M_1M_2)_{\rho^{(2)}}}  \notag
\end{align}
which gives the result \eqref{eqn17} by invoking \eqref{eqn15}, and in the third line we used $H(\rho^{(2)}_{S' M_1M_2})=H(\rho^{(2)}_{S M'_1M'_2})=H(\rho^{(2)}_{S M_1M_2})$.
\end{proof}

We note that Cor.~\ref{thmMixed2505} and Cor.~\ref{cor127}, respectively, imply Thm.~2 and Cor.~4 from the main text by setting $\rho^{(0)}_{M_1}=\dya{0}$ and $\rho^{(0)}_{M_2}=\dya{0}$.

\section{Multiple measurements}

Here we extend our main result to the case of arbitrarily many measurements, i.e., we prove Eq.~(11) from the main text. Suppose that system $S$, initially at time $t_0$ in state $\rho^{(0)}_S$, interacts sequentially with $n$ measurement devices, which each initially start in state $\ket{0}$. Recall from the main text that time $t_n$ corresponds to the time immediately after the $n$-th measurement device $M_n$, which measures observable $X^n=\{[X^n_j]\}$ of $S$, has interacted with $S$. We denote the entanglement at time $t_n$, between the system $S$ and the measurement devices $M_1\dots M_n$, as $E(X^1,\dots,X^n)$. We first provide a more mathematically detailed proof of Eq.~(10).
\begin{lemma}
\label{thm1}
Consider any entanglement measure $E$ that is either non-increasing under LOCC or is non-increasing on average under LOCC. Then
\begin{equation}
\label{eqn19}
E(X^1,\dots,X^n) \geq E(X^1,\dots,X^{n-1})
\end{equation}
\end{lemma}
\begin{proof}
The interaction of $S$ with the $n$-th measurement device $M_n$ can be written as an isometry, $V^n: \HC_S \to \HC_{SM_n}$,  as follows:
\begin{align}
\label{eqn20}V^n & = \sum_j [X^n_j] \ot \ket{j}\\
\label{eqn21}& = (1/\sqrt{\ddd }) \sum_k U^n_k \ot \ket{q_k}
\end{align}
where $\{\ket{j}\}$ is the computational basis on $M_n$. The second line rewrites things in terms of the $\{\ket{q_k}\}$ basis, which is related to $\{\ket{j}\}$ by the Fourier transform, and $U^n_k:= \sum_j \om^{jk} [X^n_j]$ is a unitary, with $\om := e^{2 \pi i /d}$. The second line makes it apparent that the interaction results in a random-unitary channel acting on $S$. Thus, if $\rho^{(n-1)}_{SM_1\ldots M_{n-1}}$ is the state at time $t_{n-1}$ then the state at $t_n$ is
\begin{align}
\rho^{(n)}_{SM_1\dots M_{n}}&=V^n \rho^{(n-1)}_{SM_1\dots M_{n-1}}(V^n)\ad \notag\\
&= (1/\ddd ) \sum_{j,k} U^n_j \rho^{(n-1)}_{SM_1\ldots M_{n-1}} (U^n_k)\ad \ot \dyad{q_j}{q_k}\notag
\end{align}
Consider the LOCC operation $\Lm$, where the party possessing $M_n$ measures it in the $\{\ket{q_k}\}$ basis, then maps $M_n$ to the $\ket{0}$ state, and then communicates the measurement result to the party possessing $S$, who undoes the appropriate local unitary (chosen from the set $\{U^n_j\}$) on $S$. The Krauss operators associated with $\Lm$ are $\Lm_j = (U^n_j)\ad \ot \dyad{0}{q_j}$, and the resulting state is 
$$\Lm(\rho^{(n)}_{SM_1\dots M_{n}}) =\sum_j  \Lm_j \rho^{(n)}_{SM_1\dots M_{n}} \Lm_j\ad=\rho^{(n-1)}_{SM_1\ldots M_{n-1}} \ot \dya{0}.$$
This is precisely the state at time $t_{n-1}$, i.e., the state at time $t_{n-1}$ can be obtained from the state at time $t_n$ by applying an LOCC operation between $M_n$ and $S$. This proves \eqref{eqn19} for measures $E$ that are non-increasing under LOCC. The proof if $E$ is non-increasing on average under LOCC follows by the same argument. This is because each member of the ensemble produced by $\Lm$ corresponds to the state at time $t_{n-1}$, i.e.,
$$\ddd  \Lm_j \rho^{(n)}_{SM_1\ldots M_{n}} \Lm_j\ad = \rho^{(n-1)}_{SM_1\ldots M_{n-1}} \ot \dya{0}. $$
Hence the average entanglement of this ensemble is just $E(X^1,\ldots,X^{n-1})$. \end{proof}

Now we are ready to prove Eq.~(11).
\begin{theorem}
\label{thm2}
Let $E$ be any of the entanglement measures listed in Sec.~\ref{app1} of the Appendix, then
\begin{equation}
\label{eqn22}
E(X_1,\dots,X_n) \geq \max_{m< n} \log \frac{1}{c_{m,m+1}}
\end{equation}
where $c_{m,m+1}:=\max_{j,k} |\ip{X^m_j}{X^{m+1}_k}|^2$.
\end{theorem}

\begin{proof}
The state at time $t_n$ falls into a class of states called premeasurement states for which the entanglement equals minus the conditional entropy \cite{ColesCollapsePRA2012}. More precisely, for a premeasurement state $\rho_{AB}$, we have $E_D(\rho_{AB}) = - H(A|B)_{\rho}$ and $E_{\fid}(\rho_{AB}) = - H_{\max}(A|B)_{\rho}$, but since the max entropy upper bounds the von Neumann entropy $H\leq H_{\max}$ \cite{TomColRen09}, this implies that $E_D \geq E_{\fid}$ for the state of interest. Since $E_D$ in turn lower bounds the other measures defined in Sec.~\ref{app1}, it suffices to prove \eqref{eqn22} for the measure $E_{\fid}$.

Note that Lemma~\ref{thm1} holds for $E_{\fid}$ since $E_{\fid}$ is non-increasing under LOCC. The proof of \eqref{eqn22} would then follow by combining Lemma~\ref{thm1} with
\begin{equation}
\label{eqn23}
E(X_1,\ldots,X_{m+1}) \geq \log \frac{1}{c_{m,m+1}},
\end{equation}
since Lemma~\ref{thm1} would allow us to apply \eqref{eqn23} iteratively to each value of $m$ ranging from $m=1$ to $m=n-1$. So we just need to prove \eqref{eqn23} for $E_{\fid}$. To do this, we apply the uncertainty relation for the min and max entropies \cite{TomRen2010} at time $t_m$, giving
\begin{align}
H_{\max}(X^{m}|M_1\ldots M_{m})_{\rho^{(m)}}&\notag\\
+ H_{\min}(X^{m+1}|S')_{\rho^{(m)}}& \geq \log \frac{1}{c_{m,m+1}},\notag
\end{align}
where we let $S'$ be a system that purifies the initial state $\rho^{(0)}_S$. At time $t_{m}$, the $X^{m}$ information is perfectly contained in the $M_{m}$ system, so $H_{\max}(X^{m}|M_1\ldots M_{m})_{\rho^{(m)}}=0$. Also, from Ref.~\cite{ColesCollapsePRA2012}, we have that $H_{\min}(X^{m+1}|S')_{\rho^{(m)}}$ is equal to $E(X_1,\ldots,X_{m+1})$ provided that the entanglement is measured here with $E_{\fid}$. Thus, the result is proven.\end{proof}

\end{document}